\documentclass[12pt,a4paper]{article}
\usepackage[utf8]{inputenc}

\usepackage{amsmath} 
\allowdisplaybreaks[4]
\usepackage{float}
\usepackage{latexsym}
\usepackage{amsfonts}       
\usepackage{amsthm} 
\usepackage{bbding}         
\usepackage{bm}             
\usepackage{graphicx}    
\usepackage[labelfont=bf]{caption}
\usepackage{array}
\usepackage{multicol}
\usepackage{multirow}
\usepackage{hhline}
\usepackage{caption}
\DeclareCaptionLabelSeparator{twospace}{\ ~}   
\usepackage{subcaption}
\usepackage{fancyvrb}       
\usepackage{dcolumn}        
\usepackage{booktabs}       
\usepackage{paralist}       

\newtheorem{thm}{Theorem}[section]
\newtheorem{pro}[thm]{Proposition}
\newtheorem{rem}[thm]{Remark}

\usepackage{chngcntr}
\counterwithin{table}{section}
\counterwithin{figure}{section}
\usepackage[title]{appendix}
\usepackage{authblk}

\newcommand{\expectation}[1]{\mathbb{E}\left[#1\right]}

\usepackage{color}
\usepackage{soul}
\definecolor{lightgreen}{rgb}{0.94, 1.0, 0.94}
\definecolor{forestgreen}{rgb}{0.13, 0.55, 0.13}

\definecolor{violet}{rgb}{0.5, 0.0, 0.5}
\definecolor{pink}{rgb}{1.0, 0.94, 0.96}

\usepackage{geometry}
\title{Sideward contact tracing in an epidemic
model with mixing groups}
\author[1*]{Dongni Zhang}
\author[2]{Martina Favero}
\affil[1]{Department of Health, Medicine and Caring Sciences, Link{\"o}ping University, 581 83 Link{\"o}ping, Sweden.}
\affil[2]{Department of Mathematics, Stockholm University, 106 91 Stockholm, Sweden.}
\affil[*]{Corresponding Author: dongni.zhang@liu.se}

\date{\today}

\begin{document}

\maketitle

\begin{abstract} 
We consider a stochastic epidemic model with sideward contact tracing. We assume that infection is driven by interactions within mixing events (gatherings of two or more individuals). 
Once an infective is diagnosed, each individual who was infected at the same event as the diagnosed individual is contact traced with some given probability. Assuming few initial infectives in a large population, the early phase of the epidemic is approximated by a branching process with sibling dependencies. To address the challenges given by the dependencies, we consider sibling groups (individuals who become infected at the same event) as macro-individuals and define a macro-branching process.  
This allows us to derive an expression for the effective macro-reproduction number which corresponds to the effective individual reproduction number and represents a threshold for the behaviour of the epidemic.
Through numerical examples, we show how the reproduction number varies with the distribution of the mixing event size, the mean size, the rate of diagnosis and the tracing probability. 
\end{abstract}

\section{Introduction}\label{sec:intro}

Contact tracing is recognized as a crucial prevention for pandemic control. Traditional contact tracing involves \textit{backward} tracing to identify who may have infected the index case and \textit{forward} tracing to find out who the index case may have infected (individuals who are diagnosed are called index cases in the contact tracing). This paper concerns an innovative tracing method known as \textit{“sideward” contact tracing}, first introduced and analysed in \cite{mancastroppa2022sideward} within the framework of simplicial activity-driven temporal networks.  
We focus on a related sideward tracing strategy which identifies individuals attending the same gathering as the index case. This mechanism is particularly relevant in the context of large gatherings where ``superspreading events" (many people get infected at once \cite{lewis2021superspreading}) usually occur. 

While \cite{mancastroppa2022sideward} makes a  significant step in the analysis of the effects of sideward tracing by focusing on an epidemic model spreading on a temporal network, we instead focus on a different framework arising from classical stochastic SIR  models (see, e.g., \cite{britton2010stochastic} for an overview). In particular, we adopt the stochastic epidemic model with mixing groups introduced in \cite{ball2022epidemic}. In most epidemic models, infections are assumed to occur between pairs of individuals (one is infectious, and the other is susceptible). In this model, instead of having pairwise interactions, individuals make contacts via attending \textit{mixing events} with at least two individuals involved. This allows for the possibility that an infective can infect more than one susceptible on a single occasion and thus provides a suitable framework for the study of sideward tracing. 

The mathematical modelling of contact tracing presents several challenges (see an overview in \cite{muller2021contact}). This paper focuses on analysing the initial behaviour of the epidemic with sideward tracing via a continuous-time branching process incorporating \textit{sibling dependencies}. In this context, an individual in the branching process corresponds to an infectious individual in the epidemic setting, while their offspring represent the secondary cases they infect and siblings are those who are infected by the same individual. Notably, sibling dependencies are not limited to epidemic settings. They are also highly relevant in biological population models, where they play a key role in understanding population dynamics \cite{olofsson1997branching}. A common example is sibling competition, which arises when siblings share limited resources \cite{roulin2012sibling}, leading to a negative correlation in the number of offspring each sibling can successfully produce.

Generally speaking, the dependencies introduced by contact tracing complicate the analysis of related branching processes. For instance, in previous studies of models with forward and/or backward contact tracing, lifetimes of siblings are co-dependent on their parents in \cite{ball2011threshold}, lifetimes of siblings and parents depend on each other in \cite{zhang2022analysing}. In this paper, the challenge arises from the lifetimes of individuals who were infected in the same mixing event being dependent on each other due to sideward contact tracing. 

To address this challenge, we consider the groups of siblings, specifically those who were born (infected) in the same birth (mixing) event, as \textit{macro-individuals},  such macro-individual behaving independently of each other. The main idea, introduced by \cite{olofsson1997branching,broberg1987sibling}, is to consider these macro-individuals as a branching process, referred to as the ``macro process", embedded into the sibling-dependent process. The key point is that macro-individuals reproduce independently, unlike single individuals, as all the dependencies are within the sibling groups. However, our definition of sibling groups differs: the sibling group in \cite{olofsson1997branching} consists of all the children with the same parent, whereas our sibling group comprises children born during the same birth event, so an individual can give birth to multiple sibling groups.

We study the macro process, which is a branching process with birth rates related to sizes of sibling groups, and we apply the standard branching process theory to derive a macro-reproduction number which represents  the epidemic threshold and corresponds to the individual reproduction number. 

The paper is structured as follows. In Section \ref{sec:model}, the epidemic model with mixing groups is defined, as in \cite{ball2022epidemic}, and a sideward contact tracing mechanism is introduced in the epidemic model. The approximation of the early stage of the epidemic, with and without contact tracing, is presented in Section \ref{sec:early_epi_approx}. In Section \ref{sec:macro_process}, the limiting process for the early epidemic with contact tracing is described as a macro branching process, which allows the derivation of macro and individual reproduction numbers in Section \ref{sec:results}. A numerical illustration of the effect of sideward tracing on the reproduction number is provided in Section \ref{sec:numerical}. Finally, Section \ref{sec:conclusion} presents the conclusion and discussion.

\section{Model}\label{sec:model}
\subsection{The epidemic model with mixing groups}\label{sec:base_model}

We describe the SIR epidemic model recently introduced and analysed in \cite{ball2022epidemic,ball2023strong,cortez2024}. Initially, there are $m_{n}$ number of infectives in an otherwise susceptible population of fixed size $n$. Differently from the standard homogeneous SIR epidemic where infections occur through pairwise interactions, infectious contacts in this model are made during temporary gatherings, referred to as ``mixing events'', which  occur at rate $n\beta$. Each mixing event involves a certain number, distributed as $C^{(n)}$ and independent of other mixing events, of individuals chosen uniformly at random from the population. The mean size of mixing events is denoted by $\mu^{(n)}_{C} = \expectation{C^{(n)}}$. Naturally, it is assumed that a mixing event involves at least two individuals and at most $n$ individuals. 

Given that a mixing event is of size $c$, i.e. it involves $c$ individuals, any given individual in the population attends the event with probability $c/n$ and any given infectious individual attending the event has probability $\pi_{c}$ of making an infectious contact with any given susceptible individual attending the event, where all such contacts occur independently. The dependence of $\pi_{c}$ on $c$ is designed to capture how event size influences the likelihood of transmission in real-world scenarios. For instance, individuals in larger gatherings may adopt more cautious behaviours, such as maintaining physical distance, which can lower the transmission risk. Conversely, smaller events tend to foster closer interactions, such as prolonged conversations, which can be more conducive to transmission.

Any susceptible individual who is contacted by at least one infective during the event becomes infected and remains infectious for a random period of time $T_{I} \sim Exp(\gamma)$ (with mean $\expectation{T_{I}}=1/\gamma$). 
Further, it is assumed that  newly infected individuals cannot transmit the disease within the same event where they were infected, i.e. mixing events are treated as instantaneous, and that there is no latency period. All the processes and random variables described above are assumed to be mutually independent.

Note that it is possible to have zero or multiple infections in one single mixing event, and that, when all mixing events involve exactly two individuals, i.e.  $C^{(n)}\equiv2$,  the above model aligns with standard homogeneously mixing SIR epidemic model with individual-to-individual rate of infection $2\beta\pi_{2}/(n-1)$ and recovery rate $\gamma$.

\subsection{Sideward contact tracing}\label{sec:model_ct}

This subsection aims to incorporate sideward contact tracing into the epidemic model described above. To this aim, we introduce a diagnosis rate $\delta$, which activates a contact tracing procedure. More precisely, infectives can be removed in three ways: 
due to natural recovery at rate $\gamma$; due to diagnosis (excluding contact tracing) at rate $\delta$, or through a contact tracing mechanism described in the following. 
For clarity, the term ``diagnosed" is used broadly here to include various causes of removal that can initiate contact tracing.  For example, an individual can be diagnosed due to the onset of symptoms or mass testing.  

We assume that when an infectious individual is diagnosed, they are removed (isolated) and the mixing event where they were infected is immediately \textit{detected}, triggering the sideward contact tracing procedure. Such individual is referred to as the \textit{index case}. Ideally, contact tracing would identify all of the events involving the index case and individuals present at the detected events, including those who infected the index case (backward tracing) and those they subsequently infected (forward tracing). However, such a comprehensive hybrid approach introduces additional dependencies between infectors and infectees, making the analysis significantly more challenging. In this paper, we focus exclusively on the novel mechanism of sideward tracing, not only because of its mathematical tractability, but also to better analyse its effect separately from the effect of the better-known mechanisms of backward and forward tracing. 
 Specifically, we assume that sideward tracing identifies only the event where the index case was infected and the individuals present at that event, while \textit{excluding} both the infector(s) and infectees.

This assumption in principle includes both individuals infected at the same event as the index case and those infected at distinct events but present at the detected event. However, in the initial phase of the epidemic, the latter type of individuals is not present because mixing events that involve at least one infective include, almost surely,  only one infective with all others being susceptible (see Section \ref{sec:early_epi_approx}). Our analysis concerns the initial phase of the epidemic, consequently, we can safely simplify the sideward tracing mechanism to tracing only those infected at the same event as the index case, while neglecting those present but infected elsewhere. 

For mathematical formulation, we define the sideward tracing procedure as follows: each individual who is infected at the same event as the index case is traced, independently with probability $p$, and tested. If found infectious, they are immediately removed to stop spreading the infection. Table \ref{tab:model_parameter} lists the important model parameters. 

Note that the epidemic model with diagnosis rate $\delta$, but without contact tracing, i.e. $p=0$, corresponds to the epidemic model defined in Section \ref{sec:base_model}, where the infectious period follows $T_{I} \sim Exp(\gamma+\delta)$. In addition, if all the mixing events involve exactly two individuals, sideward tracing has no effect, since no one will be contact traced. 

\begin{table}[htb!]
    \centering
        \caption{Key quantities related to the epidemic model}
         \label{tab:model_parameter}
         \renewcommand\arraystretch{1.2}
    \begin{tabular}{|l|l|}
    \hline
       Parameter  &  Description\\
       \hline
       $n$  & population size\\
       \hline
       $n\beta$  & rate of mixing events\\
       \hline
       $\gamma$  & rate of natural recovery\\
       \hline
       $\delta$  & rate of diagnosis\\
       \hline
       $p$  & probability of being contact traced in a detected event\\
       \hline
       $\pi_{c}$ & infection probability within a mixing event of size $c$ \\ 
       \hline
       $C^{(n)},\mu^{(n)}_{C}$ & size of a mixing event and its expected value \\
       \hline
    \end{tabular}
\end{table}
\noindent

\section{Early epidemic approximation}\label{sec:early_epi_approx}

In this section, first, we explain heuristically how the early stages of the epidemic involving mixing groups can be approximated by a branching process; see \cite{ball2023strong} for a more detailed explanation and rigorous proof. Furthermore, we describe how sideward contact tracing modifies the branching process by modifying the lifespan of individuals and by introducing dependencies between siblings who are born in the same birth event (corresponding to individuals who are infected during the same mixing event). The dependencies lead to challenges in the analysis of the reproduction number, which are highlighted here and addressed in the next section.  

We consider the early phase of an epidemic with mixing groups in a large population of size $n$ with a few initial infectives, i.e.   $m_{n}=m$ for all sufficiently large $n$. 
We make asymptotic assumptions that the events size $C^{(n)}$ converges in distribution to $C$, as $n \to \infty$, where $C$ has probability distribution $p_{C}(c):=\mathbb{P}(C=c), c =2,3,...$, with finite mean $\mu_{C}$ satisfying $\mu^{(n)}_{C} \to \mu_{C}$, as $n \to \infty$.

Since mixing events are formed by choosing individuals uniformly at random from the large population and we focus on the beginning of an epidemic, with a probability close to $1$,  each mixing event that involves at least one infective consists of only one infective and all others being susceptible. Further, note that an event of size $c$ includes a given typical infective with probability $c/n$ and recall that mixing events occur at rate $n\beta$. Consequently, mixing events involving one typical infective occur at rate 
$$\sum_{c=2}^{\infty}n\beta \frac{c}{n} \mathbb{P}(C^{(n)}=c)= \beta \mu^{(n)}_{C} \to\beta \mu_C, \quad \text{as } n\to\infty.$$
In addition,  the size of a mixing group involving a typical infective is the size-biased version of $C^{(n)}$, 
which converges to the size-biased version of $C$,
denoted by $\Tilde{C}$, with probability distribution
$$p_{\Tilde{C}}(c):=\mathbb{P}(\Tilde{C}=c)=\frac{c p_{C}(c)}{\mu_{C}}, \qquad (c=2,3,...).$$

Under the above asymptotic assumptions and additional integrability conditions, Theorem 3.1 in \cite{ball2023strong} proves that, as the population size $n \to \infty$, the number of infectives in the early stages of the epidemic with mixing groups, without considering sideward contact tracing, is approximated by a branching process $\mathcal{B}$ which we describe in the following. 

There are $m$ ancestors in the branching process $\mathcal{B}$. Alive individuals in $\mathcal{B}$ correspond to infectious individuals and a \textit{birth event} corresponds to a mixing event containing one single infective in an otherwise susceptible group in the epidemic. Once born, an individual has lifetime distribution $T_{I} \sim Exp(\gamma)$, during which they give birth at rate $\beta \mu_{C}$. It follows that the number of birth events produced by one typical individual during their lifetime is geometrically distributed as $G$, with 
\begin{equation}
\label{eq:G_dist}
	\mathbb{P}(G=k)=\left(\frac{\beta\mu_{C}}{\beta\mu_{C}+\gamma}\right)^k \frac{\gamma}{\beta\mu_{C}+\gamma}, \qquad (k=0,1,...)
\end{equation}
which has mean 
\begin{equation}\label{eq:EG}
	\expectation{G}= \frac{\beta\mu_{C}}{\gamma}.
\end{equation}

Denote by $\Tilde{Z}_{i}, i=1,\dots, G,$ the number of offspring produced at $i-$th birth event, which are i.i.d. random variables, independent of $G$ and equal in distribution to $\tilde{Z}$, described in the following. Given that the size of a mixing group is equal to $c$ which happens with probability $p_{\Tilde{C}}(c)$,  there are $c-1$ susceptibles in the group, each infected with probability $\pi_{c}$ independently. Thus, the number of individuals infected at an event of size $c$, the number of offspring $\Tilde{Z}|\Tilde{C}=c$ in a birth event of size $c$,  follows a binomial distribution $Bin(c-1,\pi_{c})$. Consequently, $\Tilde{Z}$ follows a mixed-binomial distribution, that is, 
	\begin{equation*}
	\Tilde{Z} \sim MixBin(\Tilde{C}-1,\pi_{\Tilde{C}}).
	\end{equation*}
Note that 
\begin{equation}\label{eq:E_Z_tilde}
    \expectation{\Tilde{Z}}=\expectation{\expectation{\Tilde{Z}|\Tilde{C}}}=\expectation{(\Tilde{C}-1)\pi_{\Tilde{C}}}= \sum_{c=2}^{\infty}(c-1)\pi_{c} p_{\Tilde{C}}(c).
\end{equation}
Finally, the total number of offspring produced by one typical individual during their lifetime is given by
\begin{equation*}
    \sum _{i=1} ^{G} \Tilde{Z}_{i},
\end{equation*}
and the \textit{basic reproduction number} in the epidemic with mixing groups, without contact tracing, corresponds to the mean number of offspring in the limiting branching process, i.e. 
\begin{equation}\label{eq:R0_general}
    R_{0}= \expectation{G} \expectation{\Tilde{Z}}
    =\frac{\beta\mu_{C}}{\gamma}\expectation{(\Tilde{C}-1)\pi_{\Tilde{C}}}= \frac{\beta}{\gamma} \sum_{c=2}^{\infty}c(c-1)\pi_{c} p_{C}(c),
\end{equation}
using  Equation (\ref{eq:EG}) and (\ref{eq:E_Z_tilde}).
A major outbreak in the epidemic is associated to the non-extinction of the approximating branching process $\mathcal{B}$ and occurs with non-zero probability if and only if $R_0>1$ \cite{ball2023strong}.

If we consider the epidemic with diagnosis rate $\delta$ but without contact tracing, the branching process remains the same except for the lifetime of individuals (corresponding to the infectious period) $T_{I} \sim Exp(\gamma+\delta)$). Thus, in this case  $R_{0}$ is given by
\begin{equation}\label{eq:R0_diagnosis}
    R_{0}= \frac{\beta\mu_{C}}{\gamma+\delta}\expectation{(\Tilde{C}-1)\pi_{\Tilde{C}}}=\frac{\beta}{\gamma+\delta}\expectation{C(C-1)\pi_{C}}.
\end{equation}


When sideward tracing is introduced, 
individuals who were infected during the same mixing event (\textit{siblings}) depend on each other since their infectious periods can be shortened if one of the others is diagnosed.
Consequently, as $n \to \infty$, the initial phase of the epidemic with sideward tracing is approximated by a different process $\mathcal{B}_{CT}$, which corresponds to the branching process $\mathcal{B}$ described above with modified lifespan distributions and  \textit{sibling dependencies}. 
The process $\mathcal{B}_{CT}$ is described as follows. 
As in $\mathcal{B}$, each alive individual in $\mathcal{B}_{CT}$ gives birth events at rate $\beta\mu_{C}$, a birth event in $\mathcal{B}_{CT}$ corresponds to a mixing event consisting of one infective in an otherwise susceptible group in the epidemic, 
and the number of offspring produced at each birth event are i.i.d. as $\tilde{Z} \sim MixBin(\Tilde{C}-1,\pi_{\Tilde{C}})$. However,  lifespan distributions in $\mathcal{B}_{CT}$ differ from those in $\mathcal{B}$.

Consider one typical individual in $\mathcal{B}_{CT}$ and their siblings, the individuals who were born in the same birth event.  Note that individuals who were born from the same parent in other birth events are \textit{not} called siblings here.
The individual dies due to either of the following events happening: natural death (natural recovery in the epidemic) at rate $\gamma$; removal by diagnosis at rate $\delta$; or given that $(i-1)$ of their siblings are currently alive, one of their siblings is removed due to diagnosis and tracing is successful which happens at rate $(i-1)\delta p$. 

The number  $G_{CT}$ of birth events produced by the typical individual in $\mathcal{B}_{CT}$ has a distribution which is not straightforward to compute, unlike the above $G$ without contact tracing. Consequently, when considering sideward contact tracing, the reproduction number becomes challenging to compute, despite the simple expression 
    \begin{equation}\label{eq:R_ind_1}
    R^{(ind)}_{e} = \expectation{\sum_{j=1}^{G_{CT}} \Tilde{Z}_{j}}
    = \expectation{G_{CT}}\expectation{ \tilde{Z}}. 
    \end{equation}
Table \ref{tab:quantities_br_process} lists the important quantities related to the process $\mathcal{B}$ and $\mathcal{B}_{CT}$.

Furthermore, because of dependencies between individuals, it is not obvious a priori, whether the reproduction number above represents a threshold for the behaviour of $\mathcal{B}_{CT}$ and of the epidemic (we will show in Section \ref{sec:R_ind} that it actually does). To address these challenges and to analyse the threshold behaviour of the epidemic, in the next section, we construct a macro branching process, embedded in $\mathcal{B}_{CT}$, by considering sibling groups as macro-individuals, which are independent.

\begin{table}[htb!]
    \centering
        \caption{Key quantities related to the limiting branching processes (large population)}
        \label{tab:quantities_br_process}
\renewcommand\arraystretch{1.2}
    \begin{tabular}{|l|l|}
    \hline
    Notation  &  Description\\
    \hline
       $\mathcal{B}$
       & approximating branching process for epidemic without \\
       &contact tracing\\
        \hline
         $C,\mu_C$ & limiting size of a mixing event and its expected value\\
       \hline
       $\Tilde{C}$ & size of a birth event produced by a typical individual\\
       \hline
        $\beta\mu_C$& rate of birth events\\
        \hline
        $\Tilde{Z}$& number of offspring in a typical birth event \\
        \hline
        $G$ & number of birth events by a typical individual 
        in $\mathcal{B}$ \\
        \hline
  $\mathcal{B}_{CT}$
       & approximating branching process with sibling dependencies \\
       &for epidemic 
       with contact tracing\\
        \hline
        $ G_{CT}$& number of birth events by a typical individual in $\mathcal{B}_{CT}$ \\
     \hline
    \end{tabular}
\end{table}

\section{The macro branching process}\label{sec:macro_process}
In the limiting process $\mathcal{B}_{CT}$, offspring born at the same birth event constitute a \textit{sibling group}. While dependency is within each sibling group, sibling groups are independent of their parents and, since each group is produced at independent birth events, sibling groups are independent of each other. 
By identifying sibling groups as ``macro-individual'', we define the \textit{macro process} $\mathcal{M}$ which is a branching process in terms of independent (macro) individuals. Figure \ref{fig:group_sibling} shows an example how the individual process and macro process are linked.  

\begin{figure}[htb!]
    \centering
    \includegraphics[width=.6\textwidth]{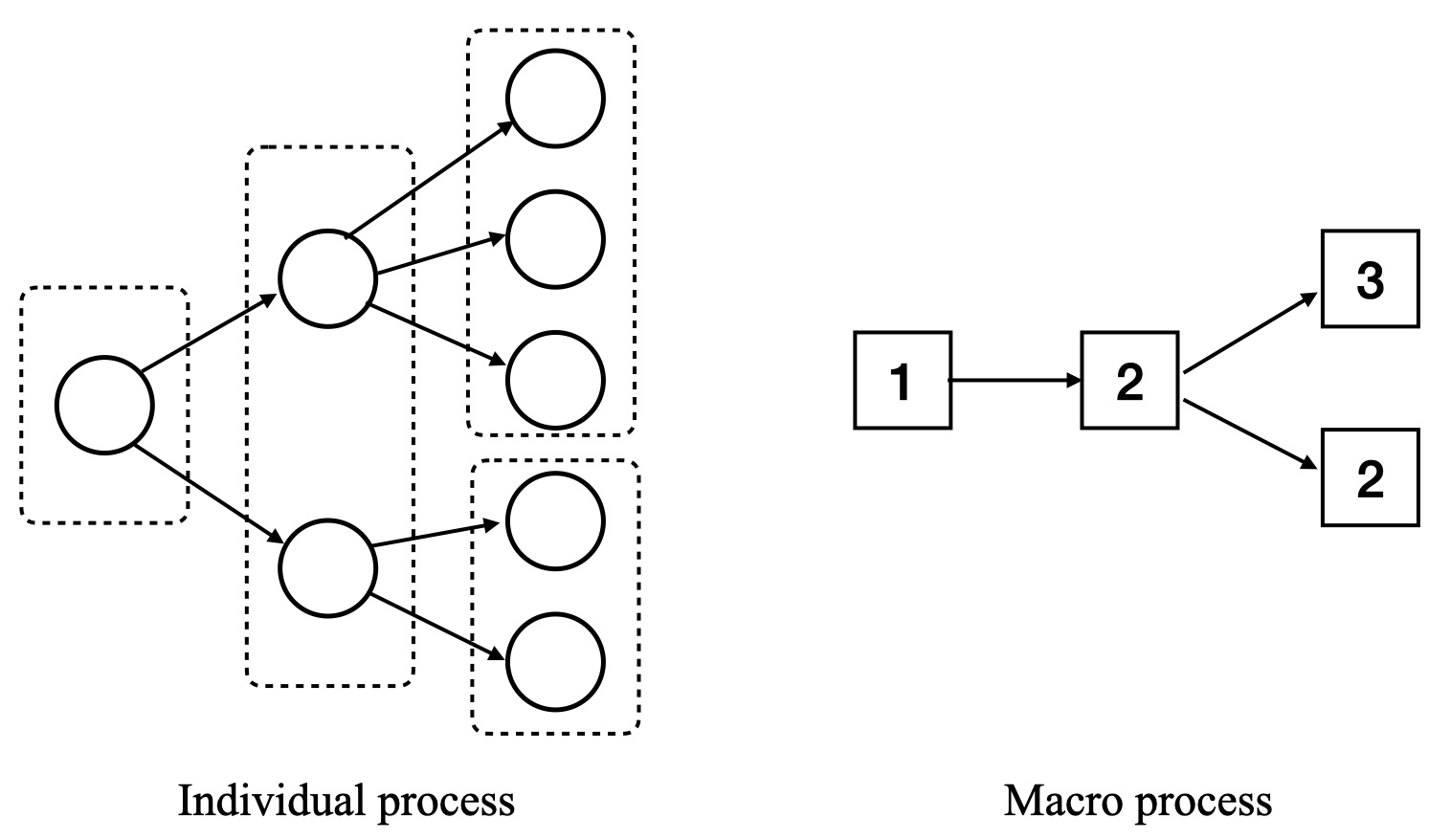}
    \caption{Illustration of an individual process $\mathcal{B}_{CT}$ and its corresponding macro process $\mathcal{M}$. On the left, circles represent individuals grouped into sibling groups, enclosed within dashed rectangles. The macro process on the right simplifies the individual process by aggregating sibling groups into macro-individuals, represented as squares labeled by the respective sibling group sizes.}
    \label{fig:group_sibling}
\end{figure}

A sibling group is born when a birth event occurs, and the number of offspring produced at the event is then the initial size of the sibling group, distributed as  
\begin{equation}\label{eq:distribution_Y_0}
    Y_{0} \overset{d}{=} \Tilde{Z} \sim MixBin(\Tilde{C}-1, \pi_{\Tilde{C}}). 
\end{equation}
Note that a sibling group can also have size zero, in that case, it certainly produces no offspring. In addition, each of the $m$ ancestors in $\mathcal{B}_{CT}$ produces independently a number of birth events distributed as $G$, since the ancestors will not be contact traced. It could happen that the $m$ ancestors produce no birth events, in this case the process $\mathcal{B}_{CT}$ dies out just after all the ancestors die and  $\mathcal{M}$ is not needed. Turning to the more interesting case where the $m$ ancestors in $\mathcal{B}_{CT}$ give $g>0$ number of birth events, the macro process $\mathcal{M}$ is thus initiated with a number $g$ of macro-individuals. Without loss of generality, in the following we assume that the macro process $\mathcal{M}$ starts with one macro-individual, and we focus on the initial conditions in Section \ref{sec:extinction}, where we study the extinction probability. 

Let $Y(t)$ be the size of a sibling group (number of alive siblings in the group) at time $t$ after the birth event, then the process $\{Y(t)\}_{t\geq 0}$ is a \textit{continuous-time Markov jump process} on the state space $\mathbb{N}$, with initial distribution $Y(0)\overset{d}{=} Y_0$ and absorbing  state $0$. Given that there is a sibling group of size $i$, several events can occur: one sibling is recovered which happens at rate $i\gamma$; one sibling is diagnosed but none of the others are traced occurring at rate $i\delta(1-p)^{i-1}$; one sibling is diagnosed and $i-1-j$ $(j=0,...,i-2)$ other siblings are traced, occurring at rate $i\delta \binom{i-1}{i-1-j}(1-p)^{j}p^{i-1-j}$. Since these occur in continuous time, the probability of two or more events happening simultaneously is zero. 

As a consequence, from a non-zero state $Y(t) = i$, two \textit{types} of jumps can occur. The first two events both result in the same type of jump: the group size is decreased by $1$ ($i\to i-1$), whereas the third one leads to another type of jump where the group size is decreased by more than 1 ($i\to j$). In conclusion, the transition rates of the process from state $i$ to $j$, denoted by $q_{i,j}$, are as follows:
\begin{equation}\label{eq:transition_rate}
  q_{i,j} = \left\{ 
\begin{array}{ll}
  i\gamma+i\delta(1-p)^{j} & \text{if } j=i-1; \\\\
  i\delta \binom{i-1}{i-j-1}(1-p)^{j}p^{i-j-1} & \text{if } j=0,...,i-2.
\end{array}\right.
\end{equation}

Each individual in a sibling group gives birth at $\beta\mu_{C}$, thus,  
at age $t$ the whole group gives birth to a new sibling group at a total rate $Y(t)\beta\mu_{C}$. The whole sibling group dies when the Markov process $\{Y(t)\}_{t\geq 0}$ reaches state 0. We can thus simply describe the macro branching process $\mathcal{M}$ as a Crump-Mode-Jagers branching process  where the birth rate of each individual at age $t$ is independent on other individuals and distributed as  $\lambda(t):=Y(t)\beta\mu_{C}$. See \cite{jagers1975} for an overview and details on Crump-Mode-Jagers branching processes.


The main advantage of defining this macro branching process is that the macro-individuals are independent and we can apply standard theory of branching processes. In particular, its reproduction number, which we will derive in the following section, corresponds to the epidemic threshold, determining whether there could be a major outbreak in the large population limit. 
The probability of a major outbreak, corresponding to the probability of non-extinction of $\mathcal{B}_{CT}$, is discussed in Section \ref{sec:extinction}.


\section{Main results}
\label{sec:results}

\subsection{Effective macro reproduction number}\label{sec:R_macro}
In this section, our interest is to derive the important quantity, $R_{e}$, which is defined as the mean number of offspring of one typical “individual” in the macro branching process $\mathcal{M}$, i.e., the expected number of sibling groups produced by one typical sibling group during its lifetime (before its size decreases to $0$). 

By standard results from the theory of branching process, the macro branching process $\mathcal{M}$ dies out with probability $1$ if $R_{e} \leq 1$; instead, if $R_{e} > 1,$ the process $\mathcal{M}$ explodes with a strictly positive probability. In Section \ref{sec:extinction}, we will show that the individual process $\mathcal{B}_{CT}$ dies out with probability 1 if and only if the macro process $\mathcal{M}$ dies out with probability $1$. Then, due to the approximation outlined in Section \ref{sec:early_epi_approx}, the epidemic may result in a major outbreak with a non-zero probability if and only if $R_{e} > 1$. We therefore refer $R_{e}$ to as the \textit{effective macro reproduction number}. 
Denoting by $H$ the number of sibling groups produced by a sibling group, 
we have
    $R_{e}=\expectation{H},$ 
which we compute more explicitly in the following. 

Consider a typical sibling group 
and let $\{Y_{k}\}_{k\in\mathbb{N}}$ be the discrete-time jump Markov chain associated to the size process $\{Y(t)\}_{t\geq 0}$. The Markov chain $\{Y_{k}\}_{k\in\mathbb{N}}$ has initial distribution given by Equation \eqref{eq:distribution_Y_0}
and absorbing state $0$. It follows from the transition rates of $\{Y(t)\}$ in Equation (\ref{eq:transition_rate}) that the transition probabilities $p_{i,j}$ (i.e., the probability of $Y_{k}$ moving from state $i$ to state $j$), $i=1,...$, are given by
\begin{equation}\label{eq:transition_prob}
  p_{i,j} = \left\{ 
\begin{array}{ll}
  \frac{\gamma+\delta(1-p)^{i-1}}{\gamma+\delta} & \text{for } j=i-1; \\\\
  \frac{\delta \binom{i-1}{i-j-1}(1-p)^{j}p^{i-j-1}}{\gamma+\delta} & \text{for } j=0,...,i-2;\\\\
  0 & \text{otherwise.}
\end{array}\right.
\end{equation}
Let $N$ be the number of jumps until $\{Y_{k}\}_{k\in\mathbb{N}}$ reaches zero, that is, 
\begin{equation}\label{eq:def_N}
    N= \inf \{k \in\mathbb{N} : Y_k=0\},
\end{equation}
and let $X_{k}, k=1,...,$ be the number of new sibling groups produced between the $(k-1)$-th jump and $k$-th jump. With this notation, we have
    \begin{equation}\label{eq:def_H}
    H= \sum_{k=1}^{N} X_{k} .
    \end{equation}
Table \ref{tab:quantities_macro_process} lists the important quantities related to the macro process $\mathcal{M}$.
    \begin{table}[htb!]
    \centering
        \caption{Key quantities related to the macro branching process}
        \label{tab:quantities_macro_process}
        \renewcommand\arraystretch{1.2}
    \begin{tabular}{|l|l|}
    \hline
    Notation  &  Description\\
    \hline
       $\mathcal{M}$ & macro branching process\\
       \hline
       $\{Y(t)\}_{t\geq 0}$, $\{Y_k\}_{k\in\mathbb{N}}$ &  
    process of sibling group size and its jump chain 
    \\ \hline
    $Y_0, \Tilde{Y}_0$& initial size of a sibling group and its size-biased version 
    \\ \hline
    $\lambda(t)=\beta\mu_{C}Y(t)$ &  birth rate by one macro-individual at age $t$ 
    \\ \hline
    $H$& number of macro-individuals generated by a typical macro-individual
    \\ \hline
    $N$& number of jumps it takes for $\{Y_k\}_{k\in\mathbb{N}}$ to reach $0$
    \\ \hline
    \end{tabular}
\end{table}

Suppose that at some point we have $Y_{k} = i \geq 1$, then the time until next jump follows an exponential time $Exp(i\gamma+i\delta)$. 
Until the next jump, each of the $i$ individuals currently alive in the group gives birth at rate $\beta\mu_{C}$, so the whole group gives birth at total rate $i\beta\mu_{C}$. Hence, $X_{k}$ is distributed as a geometric random variable $X$  with success probability $\frac{i\gamma+i\delta }{i\beta\mu_{C}+i\gamma+i\delta}=\frac{\gamma+\delta }{\beta\mu_{C}+\gamma+\delta}$, and mean 
\begin{equation}\label{eq:EX}
    \expectation{X}=\frac{i\beta\mu_{C}}{i\gamma+i\delta}=\frac{\beta\mu_{C}}{\gamma+\delta}. 
\end{equation}  
Crucially, $X_{k}$ is independent of the current size $i$, thus we have a sequence of i.i.d random variables; and, naturally, it is also independent of $N$. It follows that 
    \begin{equation}
    \label{eq:E_H}
    R_e= \expectation{H} = \expectation{\sum_{k=1}^{N} X_{k}}
    =  \expectation{X}\expectation{N}
    = \frac{\beta\mu_{C}}{\gamma+\delta}\expectation{N}.
    \end{equation}
    
To derive $\expectation{N}$, we first condition on the initial size of the sibling group, that is, 
    \begin{equation}
    \label{eq:N_sum}
    \expectation{N}
    = \sum_{i=0}^{\infty} \expectation{N|Y_{0}=i}\mathbb{P}[Y_{0}=i]
    = \sum_{i=0}^{\infty} m_{i0} \mathbb{P}[Y_{0}=i],
    \end{equation}
where  $m_{i0}:= \expectation{N|Y_{0}=i}$ is the expected number of jumps for $\{Y_{k}\}_{k\in\mathbb{N}}$ to reach state $0$,  starting from state $i$,  $i=0,1,...$ (naturally $m_{00}=0$). 
Conditioning on the first jump from a positive $i$ to $j=0,1,...,i-1$ (which happens with probability $p_{i,j}$), the expected number of jumps from $i$ to $0$ is $1+m_{j0}$ (one single jump from $i$ to $j$ plus the expected number of jumps from $j$ to $0$). Therefore, $m_{i0},i=1,...,$ is determined by the following recursive relation
\begin{equation}\label{eq:m_k0}
  m_{i0} = \sum_{j=0}^{i-1} p_{i,j} (1+m_{j0}) = 1+\sum_{j=1}^{i-1} p_{i,j} m_{j0}, 
\end{equation}
where $p_{i,j}$ are given by Equation (\ref{eq:transition_prob}). 

Using Equation (\ref{eq:E_H}) and (\ref{eq:N_sum}), the effective macro reproduction number is given by
\begin{equation*}
    R_{e}=\frac{\beta\mu_{C}}{\gamma+\delta}\sum_{i=0}^{\infty} m_{i0}\mathbb{P}(Y_{0}=i)
\end{equation*}
with $m_{i0}$ in Equation (\ref{eq:m_k0}) and   $Y_{0}\sim MixBin(\Tilde{C}-1, \pi_{\Tilde{C}})$. We have thus proved the following Theorem. 
\begin{thm}\label{thm: R_e}
The macro reproduction number $R_{e}$ for the macro branching process $\mathcal{M}$ is given by
\begin{equation}\label{eq:R_e_explicit}
    R_{e}= \frac{\beta}{\gamma+\delta} \sum_{c=2}^{\infty}c\mathbb{P}(C=c)
    \sum_{i=1}^{c-1} m_{i0} \binom{c-1}{i}\pi_{c}^{i}(1-\pi_{c})^{c-i-1},
\end{equation}
where $m_{i0}$ is given by Equation (\ref{eq:m_k0}).
\end{thm}

Another possible approach to compute $R_e$ from Equation \eqref{eq:E_H} consists of numerically approximating the expected number of steps, $\expectation{N}$, through a Monte Carlo integration based on simulating the Markov chain $\{Y_{k}\}_{k\in\mathbb{N}}$.

\subsection{Effective individual reproduction number}\label{sec:R_ind}
While it is sufficient to use the  effective  macro reproduction number $R_{e}$ to analyse the threshold behaviour of the early epidemic, as shown in the next section,  for completeness, in this section we show  that  $R_{e}$ corresponds to the \textit{effective individual reproduction number} $R^{(ind)}_{e}$ of Equation \eqref{eq:R_ind_1}. 
\begin{pro}\label{pro:ind_reprod_number}
Let $R_{e}$ be the macro reproduction number of $\mathcal{M}$ given by Equation \eqref{eq:R_e_explicit}, and let $R^{(ind)}_{e}$ be the individual reproduction number of $\mathcal{B}_{CT}$ given by Equation \eqref{eq:R_ind_1}. Then, we have 
\begin{equation}
    R^{(ind)}_{e}= R^{}_{e}.
\end{equation}
\end{pro}

Such correspondence is trivial in the absence of contact tracing, as shown in the following remark. 
\begin{rem}
In the situation when $p=0$, the size of sibling group decreases by one at each jump (due to one individual naturally recovering or being diagnosed) and we have the transition probability $p_{i,i-1}=1$.  Therefore,  $m_{i0}(p=0)=i, $ and $\expectation{N}=\expectation{Y_0}$.
It follows with Equation (\ref{eq:E_H}) that 
\begin{equation*}
    R_{e}(p=0)= \expectation{X} \expectation{N}
    =\frac{\beta\mu_{C}}{\gamma+\delta}\expectation{Y_{0}}.
\end{equation*}
Since $Y_{0} \overset{d}{=} \Tilde{Z}$, the expression above corresponds to Equation \eqref{eq:R0_diagnosis}, thus  
\begin{equation*}
    R_{e}(p=0)= R^{(ind)}_{e} (p=0).
\end{equation*}
\end{rem}

\begin{proof}[Proof of Proposition \ref{pro:ind_reprod_number}]

In general, when $0\leq p\leq1$, $R^{(ind)}_{e}$ is given by Equation \eqref{eq:R_ind_1}. In order to show $R^{(ind)}_{e}= R^{}_{e}$, it is thus enough to show that  
    \begin{equation}
    \label{eq:correspondence_ind_macro_R}
     \expectation{G_{CT}}\expectation{\Tilde{Z}} 
    = \expectation{H} , 
    \end{equation}
recalling that $G_{CT}$ and $\Tilde{Z}$ are respectively the number of birth events and the number of offspring in one birth event generated by a typical individual in  $\mathcal{B}_{CT}$; and $H$ is the number of macro-individuals (sibling groups) generated by a typical macro-individual in  $\mathcal{M}$.

The initial size of a sibling group containing the typical individual we are considering is  distributed as $\tilde{Y}_{0}$, the size-biased version of $Y_{0}$, i.e.  
    \begin{equation}\label{eq:prob_tilde_Y}
  \mathbb{P}(\Tilde{Y}_{0}=i)
  = \frac{i \mathbb{P}(Y_{0}=i)}{\expectation{Y_{0}}}, \quad (i =0,1,2,...).
    \end{equation}
Furthermore, each of the individual who belongs to a sibling group  with initial size $i$, produces, on average, $\expectation{G_{CT}|\Tilde{Y_{0}}=i}$ number of sibling groups (despite being dependent, they are identically distributed). Thus the whole sibling group of initial size $i$ generates $i \expectation{G_{CT}|\Tilde{Y_{0}}=i}$ number of sibling groups on average. This implies the following relation
\begin{equation}
    i \expectation{G_{CT}|\Tilde{Y_{0}}=i} = \expectation{H|{Y_{0}}=i}. 
\end{equation}
Consequently, the expectation of $G_{CT}$ is given by  
\begin{align*}
    \expectation{G_{CT}}=\sum_{i=0}^{\infty}  \expectation{G_{CT}|\Tilde{Y_{0}}=i}  \mathbb{P}(\Tilde{Y_{0}}=i)&= \sum_{i=1}^{\infty} \frac{\expectation{H|{Y_{0}}=i}}{i} \frac{i \mathbb{P}(Y_{0}=i)}{\expectation{Y_{0}}}
    = \frac{\expectation{H}}{\expectation{Y_{0}}}  
    .
\end{align*}
Therefore, since $\expectation{Y_{0}}=\expectation{\tilde{Z}}$, the expression above proves Equation \eqref{eq:correspondence_ind_macro_R}, and finally that
\begin{equation}
    R^{(ind)}_{e} = \expectation{{G_{CT}}}\expectation{\Tilde{Z}}
    = \frac{\expectation{H}}{\expectation{Y_{0}}} \expectation{ \Tilde{Z}}
   = \expectation{H}= R_{e}.
\end{equation}

\end{proof}

In conclusion, the individual reproduction number $R^{(ind)}_{e}$ has the same expression as the macro reproduction number $R_e$ and hence, despite the dependencies between individuals, it inherits from $R_e$ the epidemic threshold property which is proven in the next section. 

\subsection{Extinction probability}
\label{sec:extinction}
 
The goal of this section is to ensure that the macro reproduction number $R_{e}$ possesses the important threshold property 
and to provide an expression for the probability of non-extinction of the branching process  $\mathcal{B}_{CT}$. As explained in Section \ref{sec:early_epi_approx}, the non-extinction of $\mathcal{B}_{CT}$ corresponds to a major outbreak in the epidemic, in the large population limit. It remains to show that the process $\mathcal{B}_{CT}$ goes extinct with probability 1 if and only if the macro process $\mathcal{M}$ goes extinct with probability 1; otherwise, the two processes $\mathcal{B}_{CT}$ and  $\mathcal{M}$ explode with strictly positive probabilities.

In the following, we obtain an expression for the extinction probability of the branching process $\mathcal{B}_{CT}$, denoted by $z_{\mathcal{B}_{CT}}$ when there is one ancestor, and equal to $z_{\mathcal{B}_{CT}}^m$ when there are $m$ ancestors. Consider now the process $\mathcal{M}$ initiated with one macro-individual and denote by $z_{\mathcal{M}}$ its extinction probability. 
By the standard theory of branching processes, 
$z_{\mathcal{M}}$  is the smallest solution in $[0,1]$ of $f_H(s)=s$, with  $f_H(s)$ being the offspring probability generating function for $\mathcal{M}$ and $H$ defined in Equation \eqref{eq:def_H}. More explicitly, 
    \begin{equation*}
     f_H(s)= \expectation{s^H}
     =\expectation{\expectation{s^{\sum_{k=1}^{N} X_{k}}\mid N }  } 
     =\expectation{f_X(s) ^{N}},
    \end{equation*}
where $N$ is defined in Equation \eqref{eq:def_N} and $$f_X(s)=\expectation{s^ X}= \frac{\gamma+\delta}{\gamma+\delta+\beta\mu_C (1-s)} $$ is the probability generating function of $X$.

To obtain the extinction probability $z_{\mathcal{B}_{CT}}$, we condition on the number of birth events (sibling groups/macro-individuals), denoted by $J$, produced by the single ancestor in $\mathcal{B}_{CT}$ being equal to $j$. We then consider the macro process $\mathcal{M}$ initiated with $j$ macro-individuals, which has extinction probability $z_{\mathcal{M}}^j$. This implies that
    \begin{equation}\label{eq:extinction_prob_step_1}
    z_{\mathcal{B}_{CT}} = \expectation{z_{\mathcal{M}}^{J}} = \sum_{j=0}^{\infty}\expectation{z_{\mathcal{M}}^{J}| J=j}\mathbb{P}(J=j)=\sum_{j=0}^{\infty}z_{\mathcal{M}}^{j}\mathbb{P}(J=j)
     \end{equation}
Since the ancestor cannot be contact traced, but can recover or be diagnosed, $J$ is geometrically distributed with success probability $(\gamma+\delta)/(\gamma+\delta+\beta\mu_{C})$ - similar to distribution of $G$ in Equation \eqref{eq:G_dist}. That is, we have 
\begin{equation*}
	\mathbb{P}(J=j)=\left(\frac{\beta\mu_{C}}{\gamma+\delta+\beta\mu_C}\right)^j \frac{\gamma+\delta}{\gamma+\delta+\beta\mu_C}, \qquad (j=0,1,...).
\end{equation*}
Embedding the equation above into Equation \eqref{eq:extinction_prob_step_1} gives 
\begin{equation}\label{eq:extinction_prob}
    z_{\mathcal{B}_{CT}} =\frac{\gamma+\delta}{\gamma+\delta+\beta\mu_C} \sum_{j=0}^{\infty} \left(\frac{z_{\mathcal{M}}\beta\mu_{C}}{\gamma+\delta+\beta\mu_C}\right)^j = \frac{\gamma+\delta}{\gamma+\delta+\beta\mu_C (1-z_{\mathcal{M}})}.
\end{equation}
Note that in the epidemic setting with $m$ initial infectives, a major outbreak occurs with probability $1-z_{\mathcal{B}_{CT}}^{m}$.

Moreover, Equation \eqref{eq:extinction_prob} indicates that $z_{\mathcal{B}_{CT}} = 1$ is equivalent to $z_{\mathcal{M}} = 1$; and the inequality  $z_{\mathcal{B}_{CT}}<1$ is equivalent to the inequality $z_{\mathcal{M}}< 1$. As a consequence, the process $\mathcal{B}_{CT}$ dies out with probability 1 if $R_{e} \leq 1$; if $R_{e} > 1$, the process $\mathcal{B}_{CT}$ starting with one ancestor explodes with a strictly positive probability $(1-z_{\mathcal{B}_{CT}})$. This confirms that the macro reproduction number $R_{e}$ has the epidemic threshold property.

We summarize the arguments above  in the following theorem. 
\begin{thm}\label{thm:extinction}
    Let $\mathcal{B}_{CT}$ be the process approximating the early stage of the epidemic defined in Section \ref{sec:early_epi_approx}, and $R_{e}$ be the macro reproduction number given by Equation \eqref{eq:R_e_explicit}, then $R_{e}$ has the threshold property for $\mathcal{B}_{CT}$. That is,  if $R_{e} \leq 1$, the process $\mathcal{B}_{CT}$ goes extinct with probability 1; if $R_{e} > 1$, the process $\mathcal{B}_{CT}$ starting with $m \geq 1$ ancestor(s) explodes with a strictly positive probability $(1-z_{\mathcal{B}_{CT}}^{m})$ and goes extinct with the complementary probability $z_{\mathcal{B}_{CT}}^{m}$, with $z_{\mathcal{B}_{CT}}$ given by Equation \eqref{eq:extinction_prob}.
    
\end{thm}

\section{Numerical results}\label{sec:numerical}

The goal of this section is to illustrate the effect of sideward contact tracing by means of some numerical examples relying on the theoretical results of the previous section. In particular, we consider two forms of infection probability: one being fixed, i.e., $\pi_{c} \equiv \pi$ and another depends on the size, defined as $\pi_{c}=\frac{2}{c}$ for all $c \geq 2$. We consider this form for two main reasons. First, in an event involving only two people, the probability of one infectious individual transmitting the infection to the other is effectively 1, as in the standard stochastic SIR model. Second, this form serves as an example to illustrate how infection probability decreases as the event size increases. In larger events, an infective is likely to interact with more people, potentially dispersing the risk of close contact with any single individual.
Notably, when $\pi_{c}={2}/{c}$, the basic reproduction number with diagnosis $R_{0}$ given by Equation \eqref{eq:R0_diagnosis} equals ${2\beta (\mu_{C} -1)}/{(\gamma+\delta)}.$ 

Additionally, we consider three different distribution of $C$ all chosen to have the same mean $\mu_{C}$ (following the choice in \cite{ball2022epidemic}): 
\begin{enumerate}[(i)]
\item $C$ follows a logarithmic distribution with
\begin{equation*}
    \mathbb{P}(C=c)= \frac{(1-\alpha)^{c}}{(-\log(\alpha) - (1-\alpha))c}, \quad (c=2,3,...)
\end{equation*}
and in this case, the mean event size $\mu_{C} = \frac{(1-\alpha)^{2}}{(-\log(\alpha) - (1-\alpha))\alpha}$.

\item $C$ follows a geometric distribution, conditioned on being larger than 2, that is,  
\begin{equation*}
    \mathbb{P}(C=c)= (1-\alpha)^{c-2}\alpha, \quad (c=2,3,...)
\end{equation*}
with $\alpha = \frac{1}{\mu_{C}-1}.$

\item  The event size is fixed: $C \equiv \mu_{C}$.
\end{enumerate}

Through this examples, we investigate how the choice of mixing group size distribution and the forms of infection probability affects the effect of sideward contact tracing by examining its impact on the reduction of the reproduction number. The reproduction number $R_e$ is computed using its analytical expression in Equation \eqref{eq:R_e_explicit}. In our calculations, to ensure computational feasibility while maintaining approximation accuracy, the infinite sum over $c$ is truncated at $500$ in all the cases of $\mu_{C}$ and distributions of $C$.


\begin{rem}\label{rem:choice_pi}
Whenever we compare scenarios corresponding to the two possible forms of $\pi_c$, i.e. $\pi_{c} \equiv \pi$ and $\pi_{c} = {2}/{c}$, 
for a fair comparison, we choose the constant $\pi$ so that the average number of births per event is the same for both forms.  
More precisely, when $\pi_{c} = {2}/{c}$,  the average number of births per event is given by $\mathbb{E}[(\Tilde{C}-1)\pi_{\Tilde{C}}]$, implying the following condition for $\pi$: 
\begin{equation*}
    \frac{\pi\mathbb{E}[C(C-1)]}{\mu_{C}} = \frac{2(\mu_{C}-1)}{\mu_{C}}.
\end{equation*}
Solving for $\pi$ leads to 
$$\pi = \frac{2(\mu_{C}-1)}{\mathbb{E}[C(C-1)]}.$$
Specifically, if $C \sim Geo(\alpha),$ we have $\pi = \frac{1}{\mu_{C}-1}$; for $C \sim Log(\alpha)$, we get $\pi = \frac{2(\mu_{C}-1)\alpha}{\mu_{C}}$ and  $C \equiv \mu_{C}$ gives $\pi = \frac{2}{\mu_{C}}.$
\end{rem}
\begin{rem}\label{rem:R_fix_size}
Note that when $C$ is fixed, the varying infection probability $\pi_{c}=\frac{2}{c}$ if $c = \mu_{c};$ $\pi_{c}=0,$ otherwise. As a consequence, when the other parameters are the same, the effective reproduction number $R_{e}$ in terms of fixed infection probability $\pi = \frac{2}{\mu_{C}}$ is the same as the reproduction number $R_{e}$ for the varying $\pi_{c}$.
\end{rem}

First, we plot the effective reproduction number $R_{e}$ as a function of the fraction of diagnosis, $\delta/(\delta+\gamma)$, and tracing probability $p$ for three distributions of $C$: logarithmic (Figure \ref{fig:heatmap_R_log}), geometric (Figure \ref{fig:heatmap_R_geo}) and being fixed (Figure \ref{fig:heatmap_R_fix}). 
In the left panels of each figure, the infection probability $\pi_{c}$ is assumed to be independent of $c$, i.e. $\pi_{c}\equiv \pi$; in the right panels, $\pi_{c} = {2}/{c}$ for all $c$. The fixed infection probability $\pi$ is chosen specifically for each distribution of $C$ to result in the same average number of 
births per event as the case when $\pi_{c} = {2}/{c}$ (see Remark \ref{rem:choice_pi}). In either of the cases, we set $\gamma=1/7$ and choose $\beta$ such that the basic reproduction number $R_{0}$ in Equation (\ref{eq:R0_general}) equals $3$. 

As shown in Figure \ref{fig:heatmap_R_log}, \ref{fig:heatmap_R_geo} and \ref{fig:heatmap_R_fix}, $R_{e}$ is monotonically decreasing in both $\delta/(\delta+\gamma)$ and $p$ as expected. Additionally, the diagnosis fraction $\delta/(\delta+\gamma)$ appears to have a stronger effect on reducing $R_{e}$ compared to tracing probability $p$ with this difference being slightly more pronounced when $\pi_{c}$ is fixed. The choice of group size distribution has a relatively weak effect on the influence of $\delta/(\delta+\gamma)$ and $p$ on reducing $R_{e}$, as the observed patterns are similar across distributions. 
Since the three figures (Figure \ref{fig:heatmap_R_log}, \ref{fig:heatmap_R_geo} and \ref{fig:heatmap_R_fix}) show similar behaviour, we report only one of them here and refer to Appendix \ref{appendix:numerical} for the other two.  

\begin{figure}
    \centering
    \includegraphics[width=\textwidth]{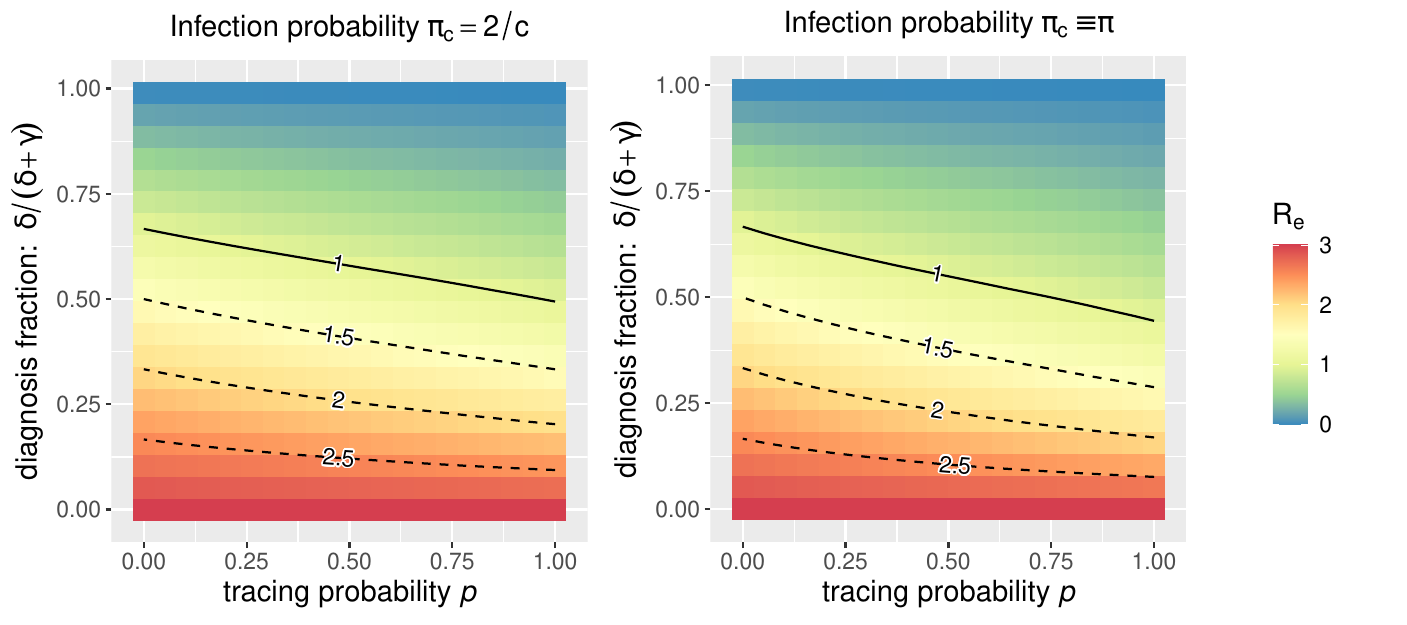}
    \caption{Heatmaps of the reproduction number $R_{e}$ as function of $\delta/(\delta+\gamma)$ in $[0,0.99]$ and $p$ in $[0,1]$. The size of mixing event $C$ follows a \textit{geometric} distribution with mean $\mu_{C}=20$, $\gamma=1/7$ and $\beta= 3/266$. In the left panel: $\pi_{c}= 2/c$ for $c \geq 2$, while $\pi_{c}\equiv 0.05$ on the right.}
    \label{fig:heatmap_R_geo}
\end{figure}

Next we focus on the reduction in the reproduction number, defined as $(R_{0}-R_{e})/R_{0}$. We note that this reduction is independent of rate $\beta$ based on the expressions in Equations \eqref{eq:R0_diagnosis} and \eqref{eq:E_H}. We further investigate how the reduction varies with the mean size of mixing events, $\mu_{C}$, the distributions of $C$ (geometric, logarithmic and fixed size) and the forms of $\pi_{C}$ ($2/c$ and fixed). Figure \ref{fig:plot_reduction_p} illustrates the reduction for $\mu_{C}=5, 10$ and $20$ while keeping $\delta=\gamma=1/7$ fixed. The fixed infection probability $\pi$ is selected for each distribution of $C$ with respective mean $\mu_{C}$ so that it gives the same average number of births in an event as for the corresponding $\pi_{c}= {2}/{c}$ case. The results suggest that when $\pi_{c}= {2}/{c}$, the distribution of $C$ does not significantly affect the reduction. In contrast, when the infection probability is fixed, the reduction becomes more sensitive to the choice of the distribution of $C$. Among the three distributions, logarithmic distribution results in the greatest reduction, followed by geometric distribution, with the fixed size showing the smallest reduction for the same mean size and tracing probability.

An explanation for the above observations is as follows. Recall that the expected number of infections per event, $\mathbb{E}[Y_{0}]$, is controlled to be the same across different distributions of $C$ with the same mean $\mu_{C}$. However, the variance $\mathrm{Var}[Y_{0}]$ depends on the distribution of $C$. A higher variance in $Y_{0}$ means that some events have higher probability of resulting in a relatively large number of infections (``super-spreading events"). These are precisely the events which are best targeted by sideward contract tracing.
When the infection probability $\pi_{c} \equiv \pi$ is fixed, $\mathrm{Var}[Y_{0}]$ increases with the variance of $C$, which is higher under the logarithmic distribution, lower under the geometric and equal to zero in the deterministic case, explaining the different effects of sideward tracing under the three considered distributions. 
On the other hand, when the infection probability takes the specific form $\pi_{c}=2/c$, a ``balancing" effect is introduced. This reduces the dependence of $\mathrm{Var}[Y_{0}]$ on the distribution of $C$, making the impact of sideward tracing more uniform across different distributions of $C$.

\begin{figure}[htb!]
\centering
\includegraphics[width=\textwidth]{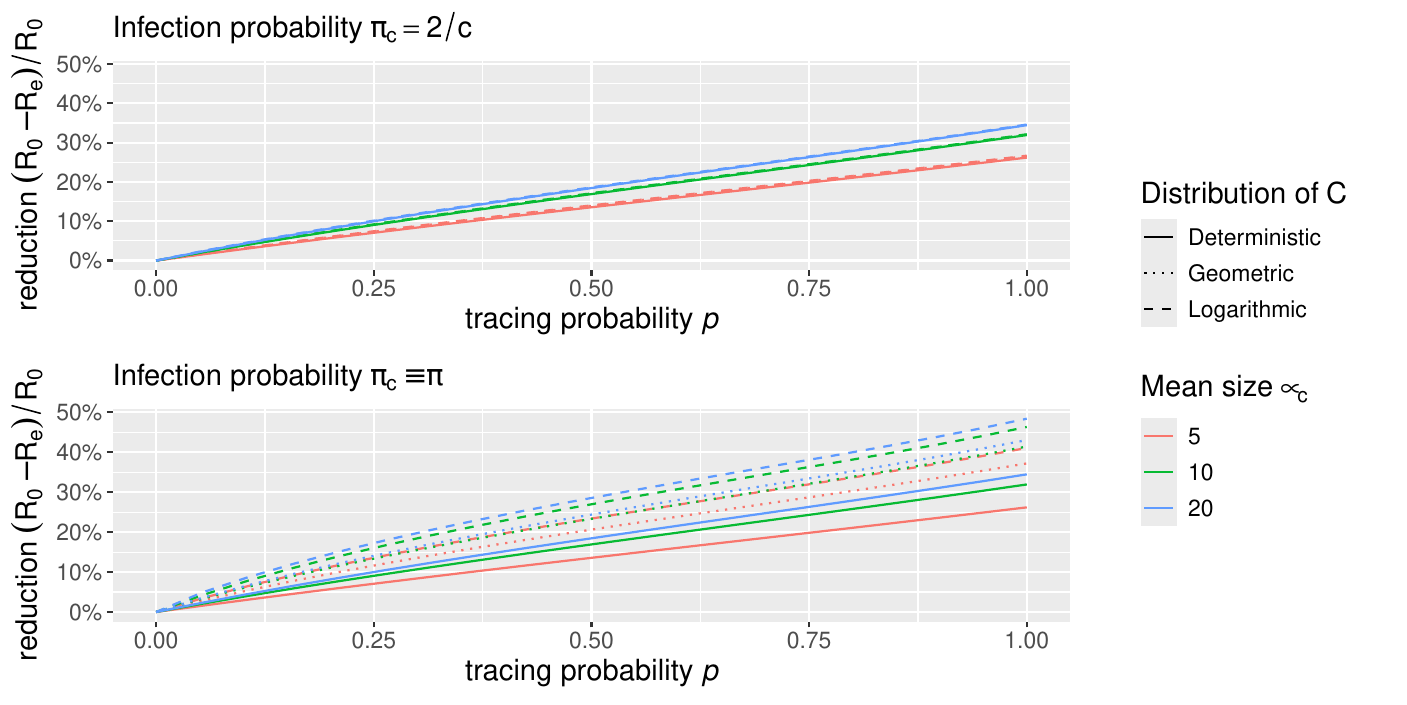}
\caption{Plot of the relative reduction $r$ as a function of the tracing probability $p$ for different mean event size $\mu_{C}=5, 10$ and $20$. We fix $\delta=\gamma=1/7$ and in the upper panel, the infection probability $\pi_{c}= {2}/{c}$ for all $c$; in the lower panel, $\pi_{c} \equiv \pi$ is chosen specifically for each distribution of $C$ to result in the same average number of infections per event as the $\pi_{c}= {2}/{c}$ case.}
\label{fig:plot_reduction_p}
\end{figure}

Crucially, we observe that for all mixing event distributions and for all forms of $\pi_{C}$, a larger $\mu_{c}$ consistently results in a greater reduction corresponding to a certain tracing probability $p$. This is due to sideward tracing being more effective when more individuals get infected during the same event. Nevertheless, this observation should not suggest that larger events are preferable, in fact, imposing limitations on the gathering size plays a bigger role in the reduction of the reproduction number. As shown in Figure \ref{fig:plot_R_same_beta}, if $\beta$ is fixed across all scenarios ($\mu_{C}=5,10$ and $20$), the case of $\mu_{C}=20$ gives a much larger $R_{0}$. In this case, even under perfect tracing ($p=1$), when the reduction is substantial,  the reproduction number $R_{e}$ remains considerably high, well above 1.

\begin{figure}[htb!]
\centering
\includegraphics[width=\textwidth]{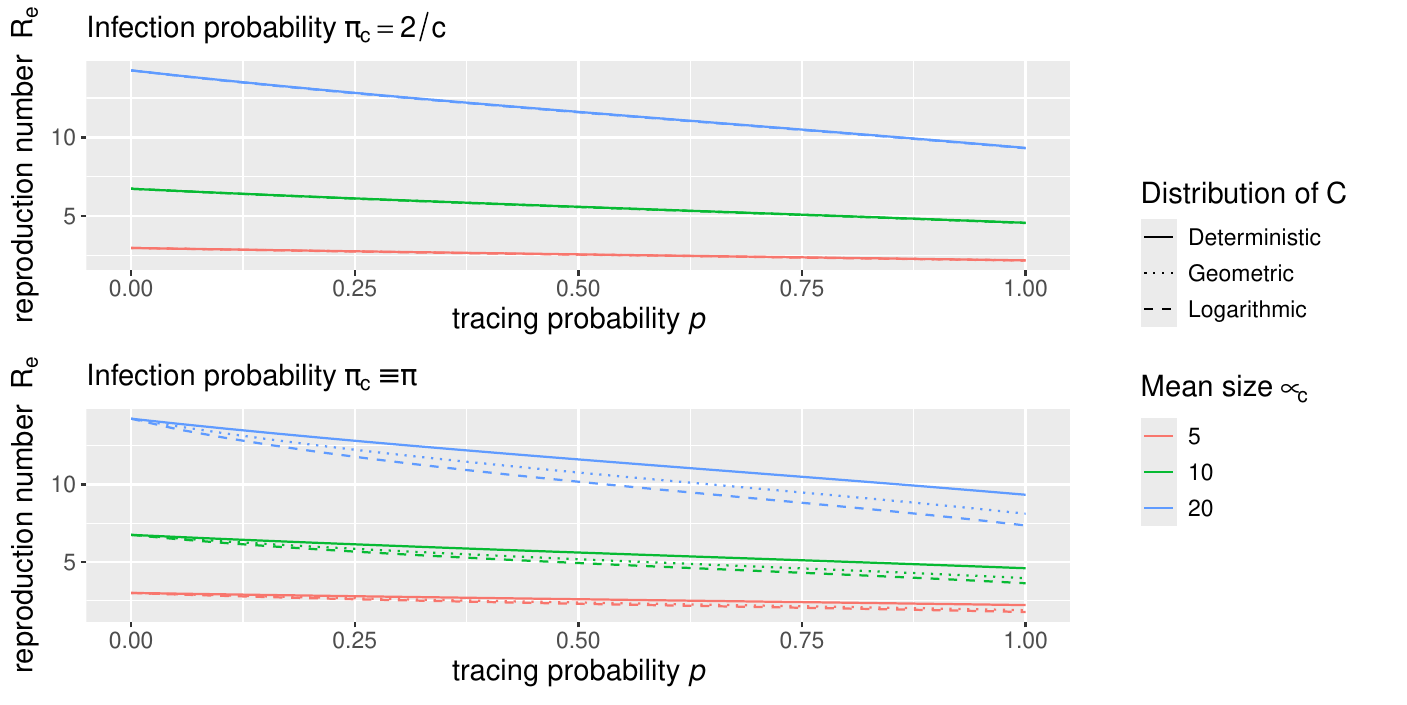}
\caption{Plot of the reproduction number $R_{e}$  as a function of the tracing probability $p$ for different mean event size $\mu_{C}=5, 10$ and $20$. We fix $\delta=\gamma=1/7$ and $\beta= 3/28$ such that $R_{0}=3$ with $\mu_{C}=5$. In the upper panel, the infection probability $\pi_{c}= {2}/{c}$ for all $c$; in the lower panel, $\pi_{c} \equiv \pi$ is chosen specifically for each distribution of $C$ to result in the same average number of infections per event as the $\pi_{c}= {2}/{c}$ case.}
\label{fig:plot_R_same_beta}
\end{figure}

\section{Conclusion and discussion}\label{sec:conclusion}
This paper investigates a novel concept: sideward contact tracing. This tracing approach is incorporated into an epidemic model that includes short-term mixing events, where multiple infections can occur at a single mixing event. In contrast to traditional tracing methods, sideward tracing aims to identify those who were infected at the same event, rather than the infector or/and the infectees of the index case. The early stage of the epidemic with sideward tracing was analysed through a branching process with sibling dependence. In particular, we treated the groups of individuals infected at the same event (group of siblings) as ``macro-individuals"; they behave independently, according to the principles of a branching process. The effective macro reproduction number $R_e$ was derived as the mean number of offspring of the macro branching process. The individual reproduction number $R^{(ind)}_{e}$, related to the original individual branching process, was also obtained. The two reproduction  numbers  have the same expression. We also expressed the probability of a major outbreak in the epidemic in terms of the non-extinction probability  of the macro branching process, which is defined through standard branching process theory.  

Using the obtained theoretical result, we performed a numerical investigation to show variations in the effect of sideward contact tracing depending on various quantities. 
Our numerical results reveal that the fraction of diagnosis has a greater impact on reducing $R_e$ compared to the tracing probability $p$, in absolute terms.
However, it should be noted that these two quantities are intertwined and hardly directly comparable. The effectiveness of tracing is inherently dependent on successful diagnosis, as tracing cannot occur if cases are not diagnosed.
Increasing the fraction of diagnosed cases presents significant challenges, particularly in the beginning of an epidemic, as it requires extensive testing efforts and the development of cost-effective tests. Similarly, increasing the fraction of traced contacts is also challenging. As studied by \cite{keeling2020efficacy}, the system can become overwhelmed when large numbers of contacts need to be traced rapidly, especially during periods of high case importation.

Furthermore, our numerical results also show that the impact of sideward tracing is more pronounced, 
resulting in a greater reduction in the reproductive number, when the size of mixing events is larger. This emphasizes the potential of sideward tracing as a control measure in connection with large gatherings. We also highlight the limitations of sideward tracing: 
without restrictions on the gathering size, the reproductive number may remain excessively high, making it unfeasible to bring it below 1 even with perfect tracing.

The theoretical results in this paper also provide technical tools for further analyses. For example, deriving the growth rate of the macro process would allow to characterize the growth rate of the epidemic as well as the generation time distribution in terms of macro individuals, see Appendix \ref{appendix:macro_process} for more details. 
Previous research \cite{favero2022} has shown that some preventive measures (as isolation, mass testing, forward and backward contact tracing) affect the generation time distribution leading to biased estimates of the reproduction number. This suggests that a similar study on the effect of sideward contact tracing on the generation time distribution and resulting biases might yield interesting results. 

In this paper, we restricted our attention to sideward contact tracing in order to  analyse its effect separately from other types of contact tracing. One promising area for further investigation arises from combining sideward with conventional forward and backward tracing procedures which can be employed instead to trace infectors and infectees. It would be interesting to explore how these three different tracing strategies can complement each other in the framework of this paper, similarly to \cite{mancastroppa2022sideward} who investigated the effectiveness of combining tracing mechanisms in the framework of temporal networks.  
For instance, in a scenario where an asymptomatic infector transmits the infection to two susceptibles during a mixing event, forward tracing may fail to identify the infectees as the infector remains undiagnosed. In such cases, sideward tracing becomes crucial if one of the infectees is diagnosed, enabling the identification of the remaining infectee. On the other hand, starting from one of the infectees, backward-forward tracing  could potentially identify the asymptomatic infector but then fail to contact trace the other infectee before recovery. Integrating these tracing strategies in our branching-processes framework is analytically challenging, particularly due to the additional dependencies it introduces between the lifespans of siblings and their parents. 

Another possible extension of our tracing model 
is to incorporate delays between the diagnosis of the index case and the notification of other siblings. The impact of such tracing delays has been extensively studied in traditional contact tracing frameworks. For instance, \cite{ball2015stochastic} analysed the impact of delays in forward tracing, while \cite{muller2016effect} investigated the effects of delays in both forward and backward tracing. Furthermore, \cite{kretzschmar2020impact} demonstrated that minimizing tracing delays and optimizing tracing coverage are crucial for the success of a contact tracing strategy.

Additionally, incorporating heterogeneity into the model could provide deeper insights. For example, \cite{kojaku2021effectiveness} highlights the effectiveness of backward contact tracing, particularly in leveraging the heterogeneity of network structures. One interesting extension to our present model can be categorizing individuals based on their levels of social activity (high, normal, or low) could reveal more about the effectiveness of sideward tracing. More socially active individuals are more likely to attend events and consequently more likely to become infected and spread the infection to others. In this case, it would be important to trace the infector as well as the siblings. 

In conclusion, this paper highlights the potential effectiveness of sideward contact tracing. We show that it can play a significant role in controlling epidemics, especially when relaxing the gathering size limitations, and we underline its limitations.

\subsection*{Acknowledgments}
We would like to thank Tom Britton for valuable discussions and helpful comments throughout this work. Additionally, we extend our gratitude to Serik Sagitov for his constructive feedback. Most of the work was conducted while D.Z. was pursuing her PhD at the Department of Mathematics, Stockholm University, 106 91 Stockholm, Sweden.

\subsection*{Funding}
D.Z. acknowledges the Swedish Research Council (grant 2020-04744) for financial support. M.F. acknowledges the Knut and Alice Wallenberg Foundation (Program for Mathematics, grant 2020.072) for financial support. 
\bibliographystyle{plain}
\bibliography{reference}
\newpage
\begin{appendices}
\section{Growth rate of the macro process $\mathcal{M}$}
\label{appendix:macro_process}
The macro-process $\mathcal{M}$ constitutes a technical framework that can be used for further analyses. 
We report here some additional technical information which we deem useful. 

Being $\lambda(t)$ the birth rate of the macro branching process, it is possible to  easily express the corresponding generation time distribution by normalising $\expectation{\lambda(t)}$ so that it integrates to $1$ (dividing by the reproduction number). 
Furthermore, the Malthusian parameter, or growth rate,  $r_{\mathcal{M}}$ can be expressed as the unique solution of 
    \begin{equation*}
    \int_0^\infty
    e^{-r_{\mathcal{M}} t} \expectation{\lambda (t)} dt =1 .
    \end{equation*}
The Malthusian parameter $r_{\mathcal{M}}$ of the macro-process $\mathcal{M}$ not only provides information about the growth of the macro-process $\mathcal{M}$, but also of the single-individuals process $\mathcal{B}_{CT}$. In fact, letting $I_{\mathcal{M}}(s)$  and $I_{\mathcal{B}_{CT}}(s)$ be the number of individuals alive at time $s$ in  $\mathcal{M}$ and $\mathcal{B}_{CT}$ respectively, we can write 
    \begin{equation*}
    I_{\mathcal{B}_{CT}}(s) =
    \sum_{i=1}^{I_{\mathcal{M}}(s)} Y_i(s-\sigma_i(s)),
    \end{equation*}
where $\sigma_i(s)$ is the time of birth of the $i^{th}$ macro-individual alive at time $s$, with $s-\sigma_i(s)$ thus being its age. 
By interpreting the size process $Y$ as a ``characteristic'' of a macro-individual, the expression above  allows us to interpret  $I_{\mathcal{B}_{CT}}(s)$ as the ``total characteristic'' of $\mathcal{M}$ at time $s$, as in \cite{jagers1984}.  
Then, the classical theory of Crump-Mode-Jagers branching processes \cite{jagers1984} implies  in the supercritical case that,  for large times, $I_{\mathcal{M}}(s) \approx e^{r_{\mathcal{M}} s} W$, with $W$ being a positive random variable; and that  $I_{\mathcal{B}_{CT}}(s)\approx e^{r_{\mathcal{M}} s} m_Y W$, with  $m_Y$ being a positive constant. 
\newpage
\section{Additional numerical results}
\label{appendix:numerical}

\begin{figure}[htb!]
\centering
\begin{subfigure}[htb!]{\textwidth}
    \centering
\includegraphics[width=\textwidth]{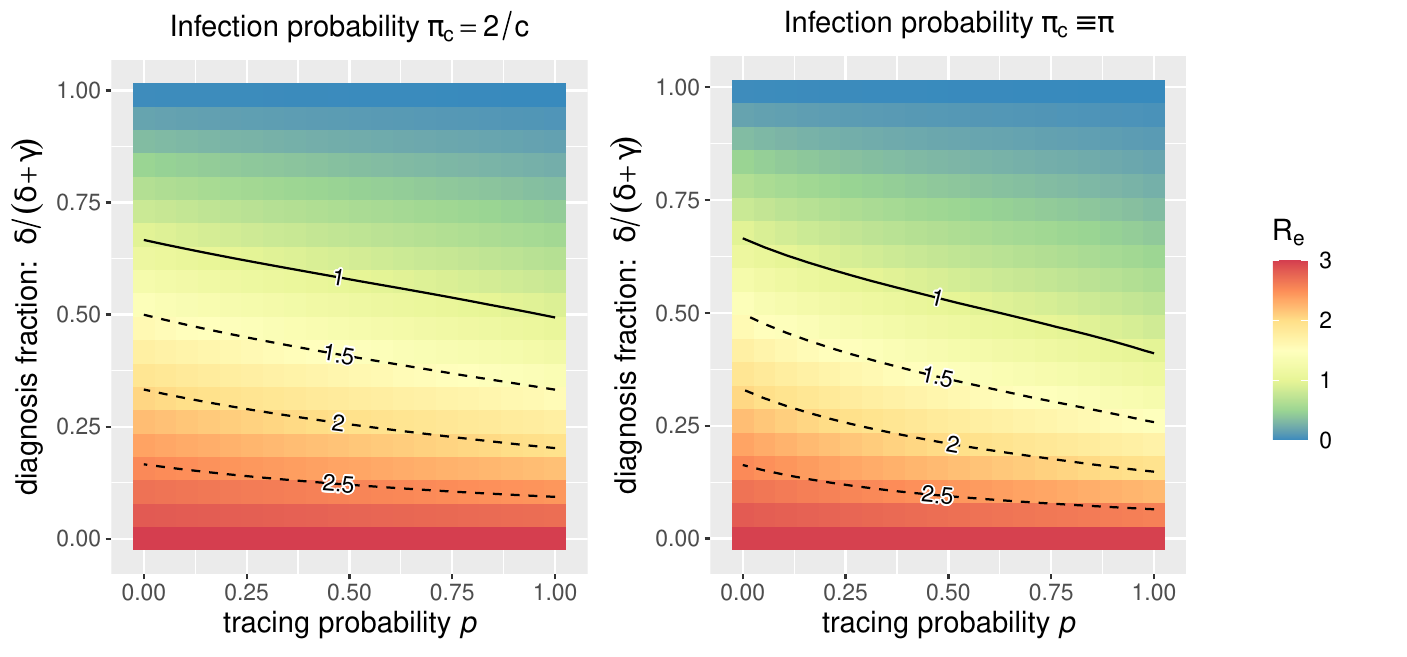}
\caption{}
\label{fig:heatmap_R_log}
\end{subfigure}
\begin{subfigure}[htb!]{\textwidth}
\centering
\includegraphics[width=.5\textwidth]{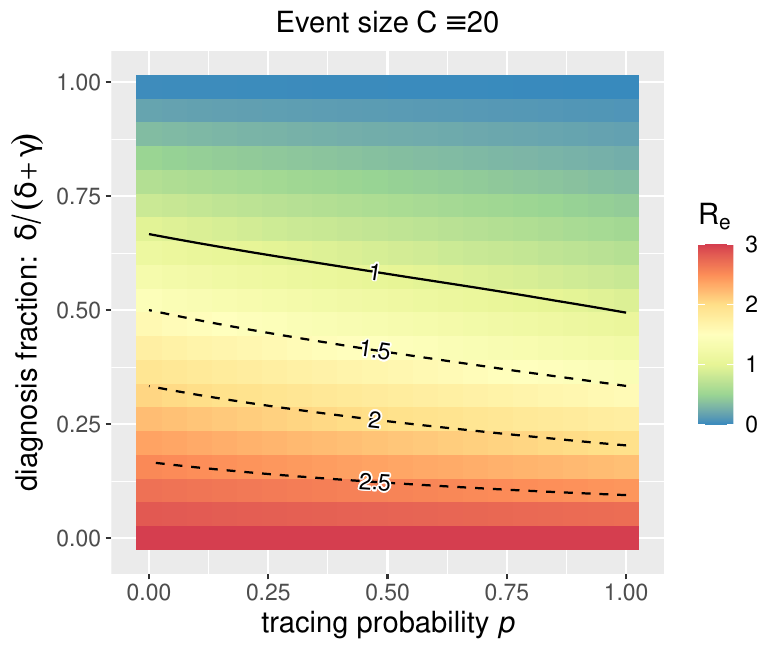}
\caption{}
\label{fig:heatmap_R_fix}
\end{subfigure}
\caption{Heatmaps of the reproduction number $R_{e}$ as function of $\delta/(\delta+\gamma)$ in $[0,0.99]$ and $p$ in $[0,1]$, where $\gamma=1/7$ and $\beta= 3/266$. The size of mixing event $C$ follows a \textit{logarithmic} distribution with mean $\mu_{C}=20$ in \ref{fig:heatmap_R_log}, where in the left panel: $\pi_{c}=2/c$ for $c \geq 2$ and $\pi_{c}\equiv 0.03$ on the right. In \ref{fig:heatmap_R_fix}, the event size is fixed: $C \equiv 20$ and hence $\pi_{c}\equiv \pi = 0.1$.}
\label{fig:heatmap_R_log_fix}
\end{figure}
\end{appendices}

\end{document}